\documentclass[a4paper, 10pt, conference]{ieeeconf}      % Use this line for a4 paper

\IEEEoverridecommandlockouts                              % This command is only needed if 
\let\labelindent\relax
\usepackage{fix-cm}
\usepackage{etex}

\usepackage{dblfloatfix}

\usepackage{nag}

\makeatletter
\@ifpackageloaded{xcolor}{}{%
\usepackage[table,x11names,dvipsnames,svgnames]{xcolor}%
}
\makeatother

\usepackage{colortbl}

\usepackage{graphicx}
\usepackage{wrapfig}

\usepackage{cite}

\usepackage{microtype}

\usepackage{array}
\usepackage{multirow}
\usepackage{booktabs}
\usepackage{makecell} %introduces \thead and \makecell

\newcommand{\myparagraph}[1]{\textbf{\emph{#1}}.}
\ifcsname labelindent\endcsname
\let\labelindent\relax
\fi
\usepackage[inline]{enumitem}

\usepackage{subfig}

\newenvironment{lenumerate}[2][]
{\begin{enumerate}[label=(#2\arabic*),leftmargin=0.2in,itemindent=0.15in,#1]}
{\end{enumerate}}

\setlist*[enumerate,1]{label={\itshape\arabic*)}}

\makeatletter
\newcommand{\paragraphswithstop}{%
\let\copyparagraph\paragraph%
\renewcommand\paragraph[1]{\copyparagraph{##1.}}%
}
\makeatother

\usepackage[framemethod=tikz]{mdframed}

\input{math}
\newcommand{\real}[1]{\mathbb{R}^{#1}{}}

\newcommand{\vct}[1]{\mathbf{#1}}

\DeclareMathOperator{\diag}{diag}

\DeclareMathOperator{\stack}{stack}

\input{utilities}
\providecommand{\mH}{\vct{H}}
\providecommand{\mI}{\vct{I}}
\providecommand{\mJ}{\vct{J}}

\providecommand{\mL}{\vct{L}}

\providecommand{\mP}{\vct{P}}

\providecommand{\mR}{\vct{R}}

\providecommand{\cA}{\mathcal{A}}

\providecommand{\cE}{\mathcal{E}}
\providecommand{\cF}{\mathcal{F}}
\providecommand{\cG}{\mathcal{G}}
\providecommand{\cH}{\mathcal{H}}
\providecommand{\cI}{\mathcal{I}}
\providecommand{\cJ}{\mathcal{J}}
\providecommand{\cK}{\mathcal{K}}
\providecommand{\cL}{\mathcal{L}}

\providecommand{\cN}{\mathcal{N}}

\providecommand{\cQ}{\mathcal{Q}}

\providecommand{\cS}{\mathcal{S}}

\providecommand{\cV}{\mathcal{V}}

\providecommand{\cX}{\mathcal{X}}

\usepackage{units}

\newcommand{\newcolorlabel}[2]{%
  \expandafter\newcommand\csname #1\endcsname[1]{%
    \colorbox{#2}{\color{white}\textsf{\textbf{##1}}}}%
}

\newcommand{\newcommenter}[2]{%
  \expandafter\newcommand\csname #1\endcsname[1]{%
    \fcolorbox{#2}{#2}{\color{white}\textsf{\textbf{#1}}}
    {\color{#2}##1}}%
  \expandafter\newcommand\csname at#1\endcsname{%
    \fcolorbox{#2}{#2}{\color{white}\textsf{\textbf{@#1}}}
    {\color{#2}}}%
  \expandafter\newcommand\csname #1hl\endcsname[2]{%
    \colorbox{#2}{\color{white}\textsf{\textbf{#1}}}\sethlcolor{Azure2}\hl{##2}~%
    \expandafter\ifx\csname commentarrow\endcsname\relax$\leftarrow$\else \commentarrow[#2]\fi~%
    {\color{#2}##1}}%
  \expandafter\newcommand\csname #1st\endcsname[2]{%
    \colorbox{#2}{\color{white}\textsf{\textbf{#1}}}\sout{##2}~%
    \expandafter\ifx\csname commentarrow\endcsname\relax$\leftarrow$\else \commentarrow[#2]\fi~%
    {\color{#2}##1}}%
}
\newcommenter{TODO}{DodgerBlue1}
\newcommenter{rtron}{Green3}

\usepackage{comment}

\usepackage{pdfcomment}

\usepackage{soul}

\usepackage[normalem]{ulem}

\usepackage{csquotes}

\providecommand{\trp}{\mathsf{T}}
\usepackage{bm}
\usepackage{mathrsfs}
\usepackage{pgfplots} 

\usetikzlibrary{plotmarks}
\usetikzlibrary{positioning}
\usetikzlibrary{arrows, automata}
\usetikzlibrary{decorations.markings}
\pgfplotsset{plot coordinates/math parser=false} 
\newlength\figureheight 
\newlength\figurewidth 
\newtheorem{conjecture}{Conjecture}

\tikzset{
  treenode/.style = {align=center, inner sep=0pt, text centered,
    font=\sffamily},
  arn_n/.style = {treenode, circle, white, font=\sffamily\bfseries, draw=black,
    fill=black, text width=1.5em},% arbre rouge noir, noeud noir
  arn_r/.style = {treenode, circle, red, draw=red, 
    text width=1.5em, very thick},% arbre rouge noir, noeud rouge
  arn_x/.style = {treenode, rectangle, draw=black,
    minimum width=0.5em, minimum height=0.5em}% arbre rouge noir, nil
}

\usetikzlibrary{shapes.misc}

\tikzset{cross/.style={cross out, draw=black, minimum size=2*(#1-\pgflinewidth), inner sep=0pt, outer sep=0pt},
cross/.default={1pt}}

\title{\LARGE \bf Bearing-Only Consensus and Formation Control under Directed Topologies}

\author{Arman Karimian and Roberto Tron% <-this % stops a space
\thanks{The authors are with the Department of Mechanical Engineering of Boston University, 
Boston, MA. E-mail: {\tt\small \{armandok,tron\}@bu.edu}.}%
}

\begin{document}

\maketitle
\thispagestyle{empty}
\pagestyle{empty}

\begin{abstract} 
%An interesting class of formation control problems is when only bearing measurements are available between the mobile agents rather than distances. In a special class of formations called rigid formations, agents become uniquely localizable up to a scale, i.e. knowing the distance between only two of the agents determines the relative positions of all agents. In such formations, the shape of the formation is determined by a scale and transformation vector.

%In many applications, it is favorable to change the formation by changing these bearing measurements without losing rigidity (e.g. entering a narrow pathway or following multiple leaders). In doing so, agents should avoid a collision or degenerate formations. More importantly, this problem becomes more complicated when the limited field of view of every agent for sensing other agents is taken into account.

%In this paper, the problems of formation control and limited ?eld of view for ?rst-order systems are addressed by a path-planning algorithm which guarantees collision avoidance and localizability in the process.
%In this paper we first prove the exponential convergence of the controller in \cite{tron2016bearing} for the pure bearing formation control problem. Later we introduce a novel controller for pure bearing formation control that is more suitable for the limited field of view problem and the directed graph.

We address the problems of bearing-only consensus and formation control, where each agent can only measure the relative bearings of its neighbors and relative distances are not available. We provide stability results for the Filippov solutions of two gradient-descent laws from non-smooth Lyapunov functions in the context of differential inclusion. For the consensus and formation control problems with undirected sensing topologies, we prove finite-time and asymptotic convergence of the proposed non-smooth gradient flows. For the directed consensus problem, we prove asymptotic convergence using a different non-smooth Lyapunov function given that the sensing graph has a globally reachable node. Finally, For the directed formation control problem we prove asymptotic convergence for directed cycles and directed acyclic graphs and also introduce a new notion of \emph{bearing persistence} which guarantees convergence to the desired bearings.

\end{abstract}

\section{INTRODUCTION}
Distributed and cooperative control of multi-agent systems using \textsl{relative bearing measurements} has gained a growing interest in recent years \cite{bishop2015distributed, zelazo2015bearing, schiano2018dynamic, trinh2018bearing}. Using bearing measurements, which are relative directions between agents, as opposed to relative positions is motivated by the use of vision-based sensors. Such sensors provide precise measurements of direction between agents while the corresponding distances are generally not known exactly. %\cite{eren2012formation}

The first problem addressed in this paper is the multi-robot rendezvous problem, which is the task of steering robots such that they eventually converge to the same location. For robots with single integrator dynamics, this problem is essentially the same as the consensus problem and has been extensively studied in the literature when difference between the states are available to agents through communication \cite{olfati2004consensus}. However, this task is not fully explored for the bearings-only case \cite{zhao2017flexible}.

The bearing-only formation control problem, whose goal is to steer a group of agents to a set of desired relative positions, is the second problem we address. In the literature, two general solutions for this task has been presented in \cite{zhao2015bearingrigid} and \cite{tron2016bearing} for single integrator dynamics. The controller given in \cite{zhao2015bearingrigid} uses an ad hoc protocol based on projector matrices while \cite{tron2016bearing} is based on minimizing a positive definite function through gradient descent. Both of these approaches, however, are limited to undirected graphs, i.e. agents should sense their relative bearings in a bidirectional manner. In \cite{zhao2015bearing}, a controller is presented for directed graphs, but relies on relative positions and the stability of the controller is not proved. In \cite{trinh2018bearing}, the controller in \cite{zhao2015bearingrigid} was extended to the Leader-First Follower structures.

The notion of bearing persistence, as was introduced in \cite{zhao2015bearing}, ensures that the desired formation is achievable in directed interaction topologies.
In addition, the notion of infinitesimal bearing rigidity (or simply rigidity) \cite{zhao2015bearingrigid} is key in guaranteeing that for a given set of bearing measurements between a group of agents, a unique class of solutions exist which only differ by a global translation, rotation and scaling of the agents' positions. While the second notion has been a subject of interest in the past years \cite{karimian2017theory, arrigoni2018bearing} , bearing persistence is fairly new and needs more attention.

The inherent discontinuous nature of bearing measurements yields differential equations with discontinuous righthand side and their proof of stability usually requires non-smooth Lyapunov functions. We present stability results in the more general context of differential inclusion for consensus and formation control problem using bearing measurements only.
%-------------------------------------------------------------------------

\noindent\myparagraph{Paper motivation}
For the consensus problem, in \cite{zhao2017flexible} a proof of stability was presented for undirected graphs, however, proof of finite-time convergence was lacking. In \cite{cortes2006finite}, a controller with bearings were proposed with finite-time convergence, however, it was limited to one dimensional space.
For the formation control problem, the existing results for directed graphs are very limited and also the definition of bearing persistence given in \cite{zhao2015bearing} is based on a controller that requires relative positions and is not compatible with a bearing-only controller.

\noindent\myparagraph{Paper contributions}
In this paper, we focus on agents with single integrator dynamics and assume that the agents have agreed on a common reference frame. Furthermore, we presume that there are no constraints on the field of view of agents and their sensors are omni-directional. Under these assumptions, for the consensus problem we extend the controller in \cite{cortes2006finite} to higher dimensions and to directed graphs. For undirected graphs, we prove that convergence happens in finite time. For the directed graphs, we only establish asymptotic stability and leave finite time convergence as a conjecture.
For the formation control problem, we prove that the controller in \cite{tron2016bearing} stabilizes directed acyclic graphs and also directed cycle graphs. We present a new definition for bearing persistence and also provide a counter example for the conjecture made in in \cite{zhao2015bearing} on stability of the given controller.
%Provide results on the equilibrium of directed formation problem
%Establish connections between the consensus and formation control problem.

%In other words, knowing the positions of two agents in a rigid framework is equivalent to knowing the positions of all agents. Whenever a formation change is needed, the final formation need to be rigid, otherwise there are infinite possible solutions, some of which may be infeasible or undesirable. Other benefits of rigidity theory in formation control are best reviewed in \cite{dorfler2010geometric}

%\cite{franchi2012decentralized}, \cite{anderson2008rigid}, \cite{eren2002framework}, \cite{krick2009stabilisation}, \cite{olfati2002distributed}, \cite{olfati2002graph}, \cite{bishop2013stabilization}, \cite{eren2007using}, \cite{hendrickx2007directed}.

%%% Local Variables:
%%% mode: latex
%%% TeX-master: "root"
%%% End:

\section{NOTATION AND PRELIMINARIES}
\subsection{General notation}
We denote the dimension of workspace by $d$. The cardinality of a set $\cS$ is given by $|\cS|$ and its convex hull and convex closure is given by $\mathrm{co}(\cS)$ and $\overline{\mathrm{co}}(\cS)$. The euclidean norm is denoted by $\|.\|$ and the Kronecker product is denoted by $\otimes$. The $d-$dimensional open and close ball centered at $\vct c$ with radius $r$ are denoted as $\mathbb{B}_d(\vct c,r)$ and $\bar{\mathbb{B}}_d(\vct c,r)$ respectively. We denote the identity matrix by $\mI_d\in\real{d\times d}$ and $\vct 1_d\in\real{d}$ denotes the column vector of all ones.
The $\stack(.)$ and $\diag(.)$ operators are used to stack column vectors vertically into a bigger column vector and square matrices diagonally into a bigger square matrix.
A \emph{projection matrix} $\mP(\vct v)$ for a vector $\vct v\in\real{d}$ is defined by:
\begin{equation}
\mP(\vct v)\triangleq \mI_d - \frac{\vct v \vct v^\trp}{\|\vct v\|^2},
\end{equation}
and is symmetric and positive semidefinite with a single zero eigenvalue that corresponds to the eigenvector $\vct v$.\\
%change d or $d_{ij}$?
\subsection{Graph Theory and Formations}
A \emph{(directed) graph} $\cG=(\cV,\cE)$ is given by a set of vertices $\cV=\{1,\dots,n\}$ connected by directional edges given by the set $\cE \subseteq \cV \times \cV$. An \emph{undirected graph} is a  graph where for every edge $(i,j)\in\cE$ the opposite edge $(j,i)$ is also in $\cE$. The complement of $\cE$ is given by $\bar{\cE}\triangleq \{ (j,i) : (i,j)\in\cE\}$. The set of neighbors of a vertex $v$ is given by $\mathcal{N}^+_v$ and $\mathcal{N}^-_v$, where the former contains the vertices to which an outgoing edge from $v$ exists and the later contains the vertices with ingoing edges to $v$. For an undirected graph, these two sets are equal and denoted as $\mathcal{N}_v$. A \emph{weighted graph} $\cG=(\cV, \cE, \cA)$  is a graph with positive weights $a_{ij}\in\real{}$ associated to every edge $(i,j)$ in $\cE$ such that $a_{ij}=a_{ji}$ if $(j,i)$ is also in $\cE$, and the \emph{adjacency matrix} $\cA=[a_{ij}]\in\real{n\times n}$ holds all the weights such that weight of edges not in $\cE$ is zero. The \emph{degree matrix} $\Delta=\diag(a_i)\in\real{n\times n}$ is a diagonal matrix with entries equal to the sum of the rows of $\cA$, i.e., $a_i=\sum_{j\in\cV} a_{ij}$.

An orientation of a graph $\cG=(\cV,\cE)$ is given by $\cG^\sigma=(\cV,\cE^\sigma)$ with $|\cE^\sigma|=m$ such that every edge $e \in \cE$ only appears in one direction in $\cE^\sigma=\{e_k\}_{k=1}^{m}$ in some arbitrary ordering.
The \emph{Oriented Incidence matrix} $\cH = [h_{ve}]\in\{\pm1,0\}^{n\times m}$ is such that for every $e_k = (i,j)\in\cE^\sigma$ we have $h_{ik}=1$ and $h_{jk}=-1$ and zero otherwise.
The \emph{Directed Oriented Incidence matrix} is given by $\cH_+ = [g_{ve}]\in\{\pm1,0\}^{n\times m}$ where 
\begin{equation}
 g_{ik} =
    \begin{cases}
      1 & e_k = (i,j)\in\cE \text{ and } e_k\in \cE^\sigma \\
      -1 &  (i,j)\in\cE \text{ and } e_k =(j,i)\in \cE^\sigma \\
      0 & \text{otherwise}.
    \end{cases}
\end{equation}
If the graph is undirected we have $\cH_+=\cH$. The \emph{Laplacian matrix} is given by $\cL \triangleq \Delta - \cA = \cH_+ \diag(w_1,\dots,w_m)\cH$ where $w_k = \max(a_{ij},a_{ji})$ for $e_k = (i,j)\in\cE^\sigma$.
In this paper, we make the standing assumption that graphs are free of self-loops (i.e. $(i,i)\notin \cE, \forall i\in\cV$), and weights are nonnegative.

A formation $\cF=(\cG,\vct{x})$ is a pairing of the vertices of $\cG$ with the vector $\vct{x}=\stack(\vct{x}_1,\dots, \vct{x}_n)\in\real{nd}$ where vertex $v$ is assigned to $\vct{x}_v\in\real{d}$ for all $v\in\cV$. For an edge $(i,j)\in\cE^\sigma$, the corresponding \emph{bearing measurement} $\vct{u}_{ij} \in \real{d}$ is defined by:
\begin{equation}
\vct{u}_{ij} \triangleq
	\begin{cases}
	\dfrac{\vct{x}_j-\vct{x}_i}{d_{ij}} & d_{ij}\neq 0\\
	\vct 0 & d_{ij}= 0
	\end{cases}
	\label{eq_bearingdef}
\end{equation}
with $d_{ij} \triangleq \|\vct{x}_j-\vct{x}_i \|$ being the Euclidean distance between vertices $i$ and $j$.

\subsection{Formation Equivalence and Bearing Rigidity}

Two formations $\cF=(\cG,\vct{x})$ and $\tilde{\cF}=(\cG,\tilde{\vct{x}})$ are:
\begin{itemize}
\item \emph{Identical} if $\vct{x}=\tilde{\vct{x}}$.
\item \emph{Congruent} if $\vct{x}=\tilde{\vct{x}}+\vct{1}_n\otimes\vct{t}$ for some $\vct{t}\in\real{d}$.
\item \emph{Similar} if $\vct{x}=s\tilde{\vct{x}}+\vct{1}_n\otimes\vct{t}$ for some $s >0$ and $\vct{t}\in\real{d}$.
\item \emph{Equivalent} if $\vct u_{ij}=\tilde{\vct u}_{ij}$ for every $(i,j)\in\cE$.
\end{itemize}
A framework $\cF$ is said to be \emph{(infinitesimally bearing) rigid} if every framework $\tilde\cF$ that is equivalent to $\cF$ is also similar to $\cF$. Intuitively, any two rigid frameworks with the same underlying graphs $\cG$ and equal bearing measurements must have a similar shape up to a translation and a scaling factor.

%We assume that $\vct{x}_v$s are distinct unless otherwise noted.

%%% Local Variables:
%%% mode: latex
%%% TeX-master: "2016-acc-rigidityRecovery"
%%% End:

\section{BEARING-ONLY CONSENSUS}
Linear consensus problems in networks with fixed undirected topologies reach consensus on a common state by minimizing the Laplacian potential which is the sum of squared differences between the states of neighboring agents \cite{olfati2004consensus}. In formation consensus application, for a formation $\cF$ with a connected and undirected graph $\cG = (\cV,\cE)$, the Laplacian potential is defined as:
\begin{equation}
\phi(\vct x) = \frac{1}{2}\sum_{\{(i,j),(j,i)\}\subseteq\cE} \|\vct x_j-\vct x_i\|^2 =\frac{1}{2} \vct x^\trp \mL \vct x
\label{eq:phi} 
\end{equation}
with $\mL = \cL\otimes \mI_d$ being the \emph{inflated Laplacian matrix} with constant unit weights for edges in $\cE$. The potential function \eqref{eq:phi} is obtained by summing the smooth edge potentials $\phi_{\{i,j\}}(\vct x_i,\vct x_j)= \frac{1}{2}d_{ij}^2$ over all edges. By setting the velocity of each agent to negative of the derivative of $\phi$ with respect to its position, we get:
\begin{equation}
\dot{\vct x}_i = -\frac{\partial \phi}{\partial \vct x_i} = \sum_{j\in\cN_i} \vct x_{j} - \vct x_{i} 
\end{equation}
or equivalently $\dot{\vct x} = -\mL \vct x$. Since $\mL$ is a constant and positive semi-definite matrix, the agents converge exponentially to their centroid and the rate of convergence is lower-bounded by the algebraic connectivity of $\cG$. Moreover, the centroid does not change at all times and agents converge to the centroid of their initial formation. However, this controller requires every agent to know its relative position with respect to all its neighbors, i.e. $\dot{\vct x}_i = \sum_{j\in\cN_i} d_{ij}\vct u_{ij}$. 

In this section we will show that only knowing the relative bearing measurements $\vct u_{ij}$ is enough for reaching consensus in finite time. We will prove that for a directed graph, consensus is reached by the controller:
\begin{equation}
\dot{\vct x}_i = \sum_{i\in\cN_i^+} \vct u_{ij},
\end{equation} if the graph has a globally reachable node.
We first begin with undirected graphs as a special case, then we will discuss the general case of directed graphs.

\subsection{Undirected graphs}

Consider the convex and non-smooth edge potential function $\varphi_{\{i,j\}}(\vct x_i,\vct x_j)= d_{ij}$, summed over all edges in $\cE$:
\begin{equation}
\varphi(\vct x) = \sum_{\{(i,j),(j,i)\}\subseteq\cE}\varphi_{\{i,j\}}
\label{eq:varphiedgetotal}
\end{equation}
By setting the velocity of each of the single-integrator agents to the opposite of the gradient of \eqref{eq:varphiedgetotal}, we obtain the following controller:
\begin{equation}
\dot{\vct x}_i = -\frac{\partial \varphi}{\partial \vct x_i} = \sum_{j\in\cN_i} \vct u_{ij}
\label{eq:consensuscontrolleri}
\end{equation}
Let $w_k = \frac{1}{d_k}$ if $d_k = \|\vct x_{j_k}-\vct x_{i_k}\|$ is not zero and $w_k = 0$ otherwise, for every $e_k=(i_k,j_k)\in\cE^\sigma$. Using variable weights $w_k$ over edges, we define the wighted laplacian matrix as $\breve{\cL}\triangleq \cH\diag(\{w_k\}_{k=1}^m)\cH^\trp$ and $\breve{\mL}\triangleq \breve{\cL}\otimes\mI_d = \mH\diag(\{w_k\mI_d\}_{k=1}^m)\mH^\trp$ where $\mH\triangleq\cH\otimes\mI_d$. Hence, the potential function in \eqref{eq:varphiedgetotal} can be written as:
\begin{equation}
\varphi = \vct x ^\trp \breve{\mL} \vct x
\end{equation}
and controller in \eqref{eq:consensuscontrolleri} is given by:
\begin{equation}
\dot{\vct x} = -\frac{\partial \varphi}{\partial \vct x} = -\breve{\mL} \vct x,
\label{eq:dynamics}
\end{equation}
or also as $\dot{\vct x}=\mH\vct u$. However, $\varphi_{\{i,j\}}$ is not differentiable when $\vct x_i = \vct x_j$. Consequently, $\varphi$ is not differentiable whenever the distance between any pair of agents connected by an edge reaches zero. In such circumstances, we pick the zero vector as a sub-gradient of $\varphi_{\{i,j\}}$ (which is always non-negative), as $\vct u_{ij}$ was defined in \eqref{eq_bearingdef}. This sudden change in magnitude of $\vct u_{ij}$ will make the right hand side of \eqref{eq:dynamics} discontinuous.

Therefore, we resort to solutions in the Filippov sense in terms of differential inclusion \cite{cortes2008discontinuous} and use non-smooth analysis to prove stability. Consider the differential equation with discontinuous right hand side:
\begin{equation}
\dot{\vct x}  = \cX(\vct x)
\label{eq:discontinuous}
\end{equation}
We consider solutions in the form of differential inclusion $\dot{\vct x} \in \cK[\cX](\vct x)$, where $\cK : \real{dn}\rightarrow2^{\real{dn}}$ is a set-valued map evaluated around $\vct x$ excluding any set $\cS$ of measure zero:
\begin{equation}
\cK[\cX](\vct x)  = \bigcap_{\delta>0}\bigcap_{\mu(\cS)=0} \overline{\mathrm{co}}\Big(\cX\big(\mathbb B_{dn}(\vct x,\delta)\setminus\cS\big)\Big).
\end{equation}
where $\mu(.)$ is the Lebesgue measure. This yields $\cX(\vct x)$ if $\cX$ is continuous at $\vct x$ or convexification of the limits of $\cX$ about points where $\cX$ is discontinuous.
Also, for a locally Lipschitz and regular function $f:\real{dn}\rightarrow\real{}$, the \emph{Clarke generalized gradient} is defined as:
\begin{equation}
\mathfrak D f(\vct x)  = \mathrm{co}\Big(\lim_{q\rightarrow+\infty} \frac{\partial}{\partial \vct x} f(\vct x_q) \; | \; \vct x_q \rightarrow \vct x, \vct x_q\notin \Omega_f \Big)
\label{eq:clarke}
\end{equation}
where $\Omega_f$ is the set of points where $f$ is not differentiable, and the \emph{set-valued Lie derivative} of $f$ is given by:
\begin{equation}
\begin{aligned}
\tilde{\mathscr L}_{\cX} f(\vct x) =& \big\{ \ell \in\real{} \; | \; \exists \vct v\in\cK[\cX](\vct x) \text{ s.t. } \\ 
&\quad\bm\zeta^\trp \vct v = \ell, \; \forall \bm\zeta\in\mathfrak D f(\vct x)\big\} \\
%=& \bigcap_{\bm\zeta\in\mathfrak D f(\vct x) }\bm\zeta^\trp \cK[\cX](\vct x),
\end{aligned}
\label{eq:lie}
\end{equation}

which can possibly be empty. Now, we introduce the LaSalle Invariance Principle for discontinuous systems:
\begin{theorem}[LaSalle Invariance Principle \cite{bacciotti1999stability}]
Let $f:\real{d}\mapsto \real{}$ be a locally Lipschitz and regular function. Let $\vct x_0\in\cS\subset\real{d}$, with $S$ compact and strongly invariant for \eqref{eq:discontinuous}. Assume that either $\max \tilde{\mathscr L}_{\cX} f(\vct x)\leq 0$ or $\tilde{\mathscr L}_{\cX} f(\vct x)=\emptyset$ for all $\vct x\in S$. Let $Z_{\cX,f}=\{\vct x\in\real{d}\; | \; 0\in\tilde{\mathscr L}_{\cX} f(\vct x)\}$. Then, any solution $\vct x:[t_0, +\infty)\mapsto\real{d}$ of \eqref{eq:discontinuous} starting from $\vct x_0$ converges to the largest weakly invariant set $M$ contained in $\overline{Z}_{\cX,f}\cap S$. Moreover, if the set $M$ is an affine collection of points, then the limit of all solutions starting at $\vct x_0$ exists and equals one of them.
%Let $\cF:\real{d}\mapsto 2^{\real{d}}$ be a non-empty and measurable locally bounded set-valued map with compact and convex values. Let $\vct x_e$ be an equilibrium of the differential inclusion $\dot{\vct x}\in\cF(\vct x)$, and let $\cD$ be an open and connected set with $\vct x_e\in\cD$. Furthermore, let $f:\real{d}\mapsto\real{}$ be a locally Lipschitz and regular function on $\cD$, such that $f(\vct x_e)=0$ and $f(\vct x)>0$ for $\vct x \notin \cD$. Then if $\max \tilde{\mathscr L}_{\cF} f(\vct x)\leq 0$ for each $\vct x\in\cD$, $\vct x_e$ is strongly stable. Moreover, if $\max \tilde{\mathscr L}_{\cF} f(\vct x)< 0$ for each $\vct x\in\cD\setminus\vct x_e$, then $\vct x_e$ is strongly asymptotically stable. If also there exists a neighborhood $\cU\subseteq\cD$
\label{thm:lasalle}
\end{theorem}
%\begin{proof}
%\cite{cortes2006finite}
%\cite{cortes2008discontinuous}
%\end{proof}

\begin{proposition}[Finite-time convergence \cite{cortes2005coordination}]
Under the same assumptions of Theorem \ref{thm:lasalle}, further assume that there exists a neighborhood $U$ of ${Z}_{\cX,f}\cap S$ in $S$ such that $\max \tilde{\mathscr L}_{\cX} f\leq \epsilon < 0$ almost everywhere on $U\setminus{Z}_{\cX,f}\cap S$. Then, any solution $\vct x: [t_0, +\infty)\mapsto\real{d}$ of \eqref{eq:discontinuous} starting at $\vct x_0\in S$ reaches ${Z}_{\cX,f}\cap S$ in finite time. Moreover, if $U=S$, then the convergence time is upper bounded by $\epsilon^{-1}(f(\vct x_0)-\min_{\vct x\in S}f(\vct x))$.
\end{proposition} 

By setting $\cX$ to be \eqref{eq:dynamics}, we see that due to $\cX$ being bounded and upper semicontinuous with nonempty, compact, and convex values, Filippov solutions of \eqref{eq:dynamics} exists.
The generalized gradient of $\varphi_{\{i,j\}}$ with respect to $\stack(\vct x_i,\vct x_j)$ is given by:
\begin{equation}
\mathfrak D \varphi_{\{i,j\}}  = 
\begin{cases}
\{ \stack(-\vct u_{ij},-\vct u_{ji})\} & d_{ij} \neq 0 \\
\{ \stack(\bm\epsilon_{ij},-\bm\epsilon_{ij})\},\; \bm\epsilon_{ij}\in\bar{\mathbb{B}}_d(\vct 0,1)  & d_{ij} = 0
\end{cases}
\end{equation}

Let $\cN_i^\bullet$ denote neighbors of $i$ whose distance to $i$ is zero. The set-valued map for $\dot{\vct x} = \cX(\vct x)$ is then given by:
\begin{equation}
\cK[\cX](\vct x) = - \mathfrak D \varphi(\vct x) = -\breve{\mL} \vct x \bm\oplus \cI
\end{equation}
where $\bm\oplus$ is the Minkowski sum and $\cI$ is the set given by:
\begin{equation}
\begin{aligned}
\cI = \{ \stack(\bm \epsilon_1,\dots,\bm \epsilon_n) \;|\; \forall i\in\cV, \,&\bm \epsilon_i\in  \bar{\mathbb{B}}_d(\vct 0,|\cN_i^\bullet|),\\
&\bm \epsilon_i+ \sum_{j\in\cN_i^\bullet}\bm \epsilon_j= \vct 0\} 
\end{aligned}
\label{eq:I}
\end{equation}
Let $\bar{\vct x} = \frac{1}{n}\sum_{i\in\cV} \vct x_i$ be the centroid of the formation. We define the disagreement vector for each agent by $\bm\delta_i = \vct x_i-\bar{\vct x}$. This can be written in the aggregate form by $\bm\delta = \mJ\vct x$, where $\mJ = (\mI_n-\frac{1}{n}\vct 1_n \vct 1_n^\trp)\otimes\mI_d$ is the matrix that removes the component of $\vct x$ in the linear subspace $\cJ=\mathrm{span}(\vct 1_n \otimes \mI_d)$. 
Now, we will show that the controller given in \eqref{eq:dynamics} is lower-bounded by the constant $\nu$ defined by:
\begin{equation}
\begin{aligned}
\nu = \min_{\vct x} \; & \|\breve\mL\vct x\| \\
 \text{s.t. } & \|\mJ\vct x\| = 1
\label{eq:minnu}
\end{aligned}
\end{equation}
Intuitively, $\nu$ depends on the topology of the graph, and similar to algebraic connectivity and is greater than zero if the graph is connected.

\begin{lemma}
$\nu >0$ if $\cG$ is connected.
\label{lemma:nu}
\end{lemma}
\begin{proof}
Notice that \eqref{eq:minnu} can be rewritten as:
\begin{equation*}
\begin{aligned}
\nu = \min_{\vct y}\;  \|&\breve\mL\vct y\| \\
 \text{s.t. }   \|&\vct y\| = 1 \\
 		&  \vct y \in \cJ^\perp 
 %		& (\vct 1_n^\trp \otimes \mI_d) \vct y = \vct 0
\end{aligned}
\end{equation*}
Since $\vct y$ belongs to the intersection of a sphere with a linear subspace, which is compact, the minimum exists. Furthermore, $\|\breve\mL\vct y\|$ is non-negative and therefore $\nu\geq 0$. We will show that $\nu\neq 0$ for connected graphs by contradiction. If $\nu$ is zero and $d_{ij}\neq 0$ for all edges in $\cE$, then $\breve\cL$ is of rank $n-1$ and $\vct y \in \mathrm{null}(\breve\mL)=\mathrm{span}(\vct 1_n\otimes\mI_d) = \cJ$. Since we assumed $\vct y \in \cJ^\perp$, then $\vct y=\vct 0$, which violates $\|\vct y\|=1$. If there are coincident adjacent agents, given the definition of a bearing vector in \eqref{eq_bearingdef}, the corresponding weight of edges connecting them is zero as if those edges were absent. Hence, the non-zero edges can be partitioned into $\kappa$ connected components ($\kappa\geq 1$) with weighted laplacians $\{\breve\mL_k\}_{k=1}^{\kappa}$ such that $\breve\mL = \diag(\breve\mL_k)$ after some permutation over nodes. Since each component is connected, $\breve\mL_k \vct y_k$ equals zero if and only if all nodes in component $k$ are coincident, where $\vct x_k$ denotes the coordinates of nodes from component $k$. Hence, $\breve\mL \vct y$ is zero if and only if all nodes of each component are coincident. Given that the nodes connected by zero-weight edges are also coincident, and these edges connect these components to form a connected graph, all the nodes need to be coincident, violating the $\|\vct y\| = 1$ condition.
\end{proof}
For the next step, we will show finite-time stability of \eqref{eq:dynamics}.
\begin{theorem}
$\max\tilde{\mathscr L}_{\cX} \varphi(\vct x)= -\|\breve\mL\vct x\|^2  \leq -\nu^2 $ 
\label{th:finite}
\end{theorem}
\begin{proof}
By definition, we have that $\mathfrak D \varphi(\vct x) = \breve{\mL} \vct x \bm\oplus \cI$ and $\cK[\cX](\vct x) = -\breve{\mL} \vct x \bm\oplus \cI$. Based on \eqref{eq:lie}, we will show the intersection of inner products of members of $\mathfrak D \varphi(\vct x)$ with $\cK[\cX](\vct x)$ is either empty or equals $-\|\breve\mL\vct x\|^2$. If none of the nodes are intersecting, $\cI$ is empty and we have $\tilde{\mathscr L}_{\cX} \varphi(\vct x)=-\|\breve\mL\vct x\|^2$. If $\cI$ is not empty, suppose exists $\bm\alpha\in\cI$ and $\ell\in\tilde{\mathscr L}_{\cX} \varphi(\vct x)$ such that:
\begin{equation*}
\bigcap_{\bm\beta\in\cI }(\breve{\mL} \vct x+\bm\alpha)^\trp(-\breve{\mL} \vct x+\bm\beta) = \ell
\end{equation*}
Since for every $\bm\beta\in\cI$, $-\bm\beta$ is also in $\cI$, then by picking the values $-\bm\alpha$ and $\bm\alpha$ for $\bm\beta$ we get $\ell=-\|\breve{\mL} \vct x+\bm\alpha\|^2$ and $\ell=-\|\breve{\mL} \vct x\|^2+\|\bm\alpha\|^2$. By setting these two terms equal and simplifying them, we have $\|\bm\alpha\|^2 + \bm\alpha^\trp\breve{\mL} \vct x=\vct 0$.

This is true only if $\bm\alpha=\vct 0$, which means $\ell=-\|\breve\mL\vct x\|^2$, or if $\bm\alpha=-\breve{\mL} \vct x$. This cannot happen since $\bm\alpha\in\cI$, its non-zero entries only correspond to agents that are intersecting and the non-zero entries of $\breve{\mL} \vct x$ correspond to agents that are not intersecting. 
Furthermore, since $\|\breve\mL(\vct x)\vct x\| = \|\breve\mL(\beta\vct x)\beta\vct x\|$ for any $\beta>0$ the magnitude of $\breve\mL\vct x$ does not change with scale and the inequality $\|\breve\mL(\vct x)\vct x\|\geq\nu$ from Lemma \ref{lemma:nu} also stands for $\|\breve\mL(\beta\vct x)\beta\vct x\|$. Hence, the proof is complete.
\end{proof}

As was shown in Theorem \ref{th:finite}, the set-valued Lie-derivative of $\varphi(\vct x)$ is upper bounded by a negative constant, which indicates that the convergence happens in \emph{finite-time}, with $t_{\text{reach}}\leq \frac{\varphi\big(\vct x(t=0)\big)}{\nu^2}$.

\begin{lemma}
The centroid of a formation under controller \eqref{eq:consensuscontrolleri} is invariant.
\label{lemma-centroid}
\end{lemma}
\begin{proof}
Let $\Xi_\kappa = \frac{1}{n}\sum_{i=1}^{n} \vct x_i^{(\kappa)} $ be the average of coordinates of all agents along dimension $\kappa\leq d$. Since $\cK[\cX](\vct x) = -\breve{\mL} \vct x \bm\oplus \cI$, for any $\bm\chi \in \cK[\cX](\vct x)$ we have that $\sum_{i=1}^{n} \bm\chi_i = \vct 0$. Therefore, $\mathfrak D\Xi_\kappa = \bigcup_{\bm\chi\in\cK[\cX](\vct x)}\sum_{i=1}^{n} \bm\chi_i^{(\kappa)}=0$ for any $\kappa\leq d$ and the proof is complete.
\end{proof}

From lemma \ref{lemma-centroid} we can see that the agents converge to the average value of their initial positions and this centroid is invariant along time.

\subsection{Directed graphs}
In the previous section, we investigated consensus for undirected graphs. In practice, however, agents may not sense the bearing vectors of their neighbors in a bidirectional manner or through communication. As we will show in this section, having bidirectional sensing information is not necessary. We model these interactions with a directed \emph{sensing graph} $\cG$, where $(i,j)\in\cE$ means that $i$ can measure $\vct u_{ij}$. 
As we showed earlier, for an undirected graph it suffices for the graph to be connected in order to reach consensus. In this section, we investigate the controller given in \eqref{eq:consensuscontrolleri} but for the directed graph $\cG$, which is:

\begin{equation}
\dot{\vct x}_i = \sum_{j\in\cN_i^+} \vct u_{ij}
\label{eq:consensuscontrolleridirected}
\end{equation}

or as $\dot{\vct x}=\mH_+\vct u$. We will show that it suffices for $\cG$ to have a globally reachable node, or equivalently, the complement of $\cG$ to have a directed spanning tree in order to reach consensus. 

\begin{assumption}
The directed graph $\cG$ has a globally reachable node.
\label{as:directed}
\end{assumption}

The intuition behind \eqref{eq:consensuscontrolleridirected} is that each agent $i$ has a private convex objective function $\varphi_i(\vct x)=\sum_{j\in\cN_i^+}d_{ij}$ which tries to minimize by moving in the direction of $-\frac{\partial \varphi_i}{\partial \vct x_i}$. The minimizer of $\varphi_i$ with respect to $\vct x_i$ is unique if $\{\vct x_j\}_{j\in\cN_i^+}$ are not collinear and is called the geometric median or Fermat point \cite{minsker2015geometric}. The geometric median is always inside the convex hull of neighbors of $i$ and hence $i$ reaches consensus with its neighbors if they all converge to the same point.

Assumption \ref{as:directed} ensures that all nodes converge to the same point determined by the globally reachable node or nodes. The globally reachable node can be unique, which is referred to as \emph{leader}, or belongs to a strongly connected component of the graph in which case all the nodes in the strongly connected component are reachable by other nodes of the graph. Leader is stationary since it has no neighbors and all other nodes converge to it. If there is more than one globally reachable node, the convergence point of the strongly connected component composed of globally reachable nodes determines the final convergence point.

In the linear consensus problem with controller $\dot{\vct x}_i = \sum_{j\in\cN_i^+} \vct x_j-\vct x_i$, the same assumption is sufficient for consensus \cite{wu2005synchronization}. Instead of sensing graphs, the convention is to use communication graphs where edges show the direction of information flow and are essentially the the reverted version of the sensing graphs by definition. For a communication graph, the assumption \ref{as:directed} is equivalent to $\bar{\cG}$ having a directed spanning tree.

First we show that the equilibrium points of \eqref{eq:consensuscontrolleridirected} are in $\cJ$. Later, we introduce the maximum distance between any pair of nodes as a Lyapunov function for \eqref{eq:consensuscontrolleridirected} and prove stability.

\begin{lemma}
Under assumption \ref{as:directed}, $\dot{\vct x} = \vct 0$ if and only if consensus is reached.
% $\dot{\vct x} = \vct 0$ $\vct 0 \in \cK[\cX](\vct x)$
\end{lemma}

\begin{proof}
If no two neighboring agents are colliding at an instance, all edge weights are positive ($w_k>0$) and $\dot{\vct x} = \breve\mL_+ \vct x$ where $\breve\mL_+ = \cL_+\otimes \mI_d$ and $\cL_+$ is the weighted Laplacian of a graph with a globally reachable node. $\cL_+$ has rank $n-1$ \cite[Lemma 2]{wu2005rayleigh} with $\vct 1_n$ being the eigenvector corresponding to the single zero eigenvalue while other eigenvalues are positive. Therefore, $\mathrm{null}(\breve\mL_+)=\cJ$  and $\dot{\vct x}$ is zero whenever $\vct x\in\cJ$ which means that agents are in consensus. If there are some coincident neighbors in formation $\cF=(\cG, \vct x)$, say $\vct x_p=\vct x_q$ for $q\in\cN_p^\bullet$,  since the weight of edges connecting coinciding agents is zero we can assume those edges (i.e. $(p,q)$) are removed. We group such nodes $p$ and all $q\in\cN_p^\bullet$ and all $r\in\cN_q^\bullet$ and so on recursively into sets $\{\cQ_i\}_{i=1}^{n^\prime}$ with $n^\prime<n$. We introduce a new formation $\cF^\prime=(\cG^\prime, \vct x^\prime)$ with $n^\prime$ vertices where node $i$ is connected to $j$ in $\cG^\prime$ if exists at least a  vertex in $\cQ_i$ connected to a vertex in $\cQ_j$ in $\cG$. Since connectivity is maintained in this transformation, $\cG^\prime$ also has a globally reachable node. We set $\vct x_i^\prime=\vct x_q$ for any $q\in\cQ_i$ and $\dot{\vct x}_i^\prime=\sum_{q\in\cQ_i} \dot{\vct x}_q$. Since nodes of $\cG^\prime$ are not coincident, $\dot{\vct x}^\prime\neq \vct 0$ which yields $\dot{\vct x}\neq \vct 0$.
\end{proof}

Now, we will show global stability of controller \eqref{eq:consensuscontrolleridirected}. 

\begin{theorem}
Controller \eqref{eq:consensuscontrolleridirected} achieves consensus under assumption \ref{as:directed}.
\label{thm:directedstability}
\end{theorem}
\begin{proof}
Take the non-smooth Lyapunov function $V(\vct x) = \max_{p,q\in\cV} \|\vct x_p-\vct x_q\|$ to be maximum euclidean distance between the nodes of $\cG$. Since $V(\vct x)=0$ means all nodes are coincident, $\vct x$ must belong to the subspace $\cJ$. Now we only need to show that $\tilde{\mathscr L}_{\cX} V(\vct x) < 0$. 
Let $p$ and $q$ be the only two nodes with maximum distance $d_{pq}$. Let $\vct e_{pq}=\frac{\vct x_q-\vct x_p}{\|\vct x_q-\vct x_p\|}$ be the unit vector pointing to $q$ from $p$. Hence, $\frac{\partial V}{\partial \vct x_p}=-\vct e_{pq}$ and $\frac{\partial V}{\partial \vct x_q}=\vct e_{pq}$ while other derivatives are zero. Unless either $p$ or $q$ is the leader, both nodes have neighbors. For any $k\in\cN_p$, we can write $\vct x_q-\vct x_p=\vct x_q-\vct x_k + \vct x_k-\vct x_p$, or equivalently $d_{pq}\vct e_{pq}=d_{pk}\vct u_{pk}+d_{kq}\vct e_{kq}$ with $d_{pk},d_{kq}< d_{pq}$. Taking a dot product of both sides with $\vct e_{pq}$, we get $\vct e_{pq}^\trp\vct u_{pk}>0$. Therefore, since $\dot{\vct x}_p = \sum_{k\in\cN_p}\vct u_{pk}$ we get $\vct e_{pq}^\trp\dot{\vct x}_p>0$. Same argument is valid for $q$ if $q$ is not the leader. Hence, $\tilde{\mathscr L}_{\cX} V(\vct x) = \vct e_{pq}^\trp(\dot{\vct x}_q-\dot{\vct x}_p)<0$. 
Now suppose there is more than a single pair of nodes with maximum distance between them, probably with some coinciding nodes. In this case, $\Omega_V$ is the set of all positions such that there exists more than one pair of nodes with maximum distance and $\mathfrak D V(\vct x)$ is the convex hull of limits of derivatives of $V(\vct x)$ as $\vct x$ is approached from outside of $\Omega_V$. Therefore, for any pair $\{p,q\}$ with maximum distance we have $\vct y\in\mathfrak D V$ such that $\vct y_p = -\vct y_q = -\vct e_{pq}$ and other entries are zero, and $\mathfrak D V$ is the convex hull of such vectors. Moreover, from the earlier argument we have $\dot{\vct x}^\trp \vct y <0$.
If none of the pairs with maximum distance are coincident with any of their neighbors, we have $\cK[\cX](\vct x)_p=\dot{\vct x}_p$ for any node $p$ from the pairs and consequently $\dot{\vct x}^\trp \bm\zeta<0$ for any $\bm\zeta\in\mathfrak D V$. Therefore, $\tilde{\mathscr L}_{\cX} V(\vct x)$ is the intersection of negative values which is either negative or empty. In the case that a node from a pair is coincident with a neighbor, say $p$ is coincident with $p^\prime\in\cN_p$ from $\{p,q\}$, then $\{p^\prime,q\}$ is also a pair with maximum distance. We have $\cK[\cX](\vct x)_p=\dot{\vct x}_p+\bm\epsilon$ for $\bm\epsilon\in\bar{\mathbb B}_d(\vct 0, 1)$ and $\cK[\cX](\vct x)_{p^\prime}=\dot{\vct x}_{p^\prime}$. In this case $\tilde{\mathscr L}_{\cX} V$ becomes the intersection over the inner product of members of two sets, and since for the pair $\{p^\prime,q\}$ the Lie derivative is negative, the intersection is again either negative or empty. Therefore, from Theorem \ref{thm:lasalle} asymptotic stability of consensus follows.
\end{proof}
%Let $\{p,q\}$ and $\{p^\prime,q^\prime\}$ be any two pair of nodes with maximum distance such that $\{p,q\}\neq\{p^\prime,q^\prime\}$. By the definition in \eqref{eq:clarke}, $\mathfrak D V(\vct x)$ is the convex hull of limit of derivatives of $V$ such that point $\vct x$ is approached from outside of $\Omega_V$. Hence $\mathfrak D V(\vct x)$ is the convex hull of two $dn$ dimensional vectors $\vct y$ and $\vct y^\prime$ such that $\vct y_p = -\vct y_q = -\vct e_{pq}$ and $\vct y_{p^\prime} = -\vct y_{q^\prime} = -\vct e_{p^\prime q^\prime}$ and other entries are zero. If nodes $p,q,p^\prime,q^\prime$ are distinct, or share a node (say $p=p^\prime$) but neither $q$ nor $q^\prime$ are coincident if one is a neighbor of the other, same argument for vectors $\vct y$ and $\vct y^\prime$ is valid and therefore $\dot{\vct x}^\trp(\beta \vct y + (1-\beta)\vct y^\prime)<0$ for $\beta\in[0,1]$. 

%\cite{bacciotti1999stability}

Theorem \ref{thm:directedstability} only establishes asymptotic stability. However, from observation it can be seen the convergence happens in finite time. A framework with a directed graph $\cG$ satisfying assumption \ref{as:directed} and with dynamics given in \eqref{eq:consensuscontrolleridirected} can be seen as a cascade system. Partitioning $\cG$ into strongly connected components, each component is seen as a subsystem. Since there is path between every subsystem to the component containing the globally reachable node(s), subsystems form a directed acyclic graph with a single leaf. Therefore, the first step in proving finite-time convergence of \eqref{eq:consensuscontrolleridirected} is to show finite-time convergence in strongly connected graphs. Here, we present a conjectured upper bound on the convergence time in strongly connected graphs.

\begin{conjecture}
In a strongly connected graph with $n$ nodes, convergence of controller \eqref{eq:consensuscontrolleridirected} happens in finite time and the convergence time is upper bounded by $\frac{l}{2n}\sec^2(\frac{\pi}{n})$ where $l$ is the sum of distances between nodes over the longest hamiltonian cycle in the initial formation at $t_0$.
\end{conjecture}

\section{BEARING-ONLY FORMATION CONTROL}
The goal of bearing-only formation control is to achieve and maintain a desired formation specified by bearings for each edge in the sensing graph using only bearing measurements, as opposed to linear formation control which requires relative positions instead of bearings.

Linear formation control problems draw advantage from the linearity of the controller $\dot{\vct x}= -\mL \vct x$ in the consensus problem. A simple change of variables leads to exponential convergence to a desired formation congruent to $\vct x^*$ by means of  $\dot{\vct x}= -\mL (\vct x-\vct x^*)$ which only differs by a constant term $\mL\vct x^*$.
In this section, we address the nonlinear formation control problem using bearings for undirected and directed sensing graphs. Similar to the linear problem, the controllers proposed are of the form $\dot{\vct x}=\bm f(\vct x)-\bm f(\vct x^*)$ and differ by a constant term $-\bm f(\vct x^*)$ compared to consensus controllers $\dot{\vct x}=\bm f(\vct x)$ introduced in the previous section. 

Specifically, we prove Lyapunov stability of Filippov solutions of the controller given in \cite{tron2016bearing} for undirected graphs and also prove cascade stability of the aforementioned controller for directed acyclic graphs. For directed cyclic graphs, we present an example which shows that the Jacobian matrix of the controller in \cite{tron2016bearing} may have eigenvalues with positive real parts. Along the same line, we present another example that shows directed bearing Laplacian matrix may have eigenvalues with negative real parts, rejecting the conjecture in \cite{zhao2015bearing}.

%\TODO{Each agent is responsible for achieving all the bearings of its outgoing edges. Single integrators. Talk about the problem in Roberto's paper, given goal bearings, converge to a formation similar to x*, given undirected graph G}

\subsection{Undirected graphs}
Given an undirected graph $\cG$, the following non-smooth and non-convex edge potential function was suggested in \cite{tron2016bearing} (as reformulated in \cite{zhao2015bearing}):
\begin{equation}
\psi_{\{i,j\}}(\vct x_i,\vct x_j,\vct u^*_{ij})= \frac{1}{2}d_{ij}\|\vct{u}_{ij} - \vct{u}_{ij}^*\|^2,
\end{equation}

which is zero only if $\vct u_{ij}$ equals to $\vct u_{ij}^*$ or if $d_{ij}$ is zero. Similar to the undirected consensus problem, summing these terms over all edges yields the following objective function:

\begin{equation}
\psi(\vct x,\vct u^*)=\sum_{\{(i,j),(j,i)\}\subseteq\cE} \psi_{\{i,j\}}(\vct x_i,\vct x_j,\vct u^*_{ij})
\label{eq:lyapbearing}
\end{equation}

By setting the velocity of each node to be the negative of the gradient of $\psi$ with respect to its position, we obtain the controller given in \cite{tron2016bearing}:

\begin{equation}
\dot{\vct x}_i = -\frac{\partial \psi}{\partial \vct x_i} = \sum_{j\in\mathcal{N}_i} \vct{u}_{ij} - \vct{u}_{ij}^*,
\label{eq:dynamicsformation}
\end{equation}
which can be written in the aggregated form as:

\begin{equation}
\dot{\vct{x}}= \vct H (\vct u - \vct u^*).
\label{eq:dynamicsformationagg}
\end{equation}

Similar to the potential function in the consensus problem, $\psi_{\{i,j\}}$ is not differentiable when $d_{ij}$ is zero and \eqref{eq:dynamicsformationagg} therefore becomes discontinuous when two agents are coliding. Denoting \eqref{eq:dynamicsformationagg} by $\cX$, the set valued map of $\cX$ is given by:

\begin{equation}
\cK[\cX](\vct x) = - \mathfrak D \psi(\vct x) = \mH (\vct u - \vct u^*) \bm\oplus \cI
\end{equation}

where $\cI$ is defined in \eqref{eq:I}. Similar to the undirected consensus problem, asymptotic stability can be established by using \eqref{eq:lyapbearing} as Lyapunov function.

\begin{proposition}
Controller \eqref{eq:dynamicsformationagg} is asymptotically stable.
\end{proposition} 
\begin{proof}
Following the proof of Theorem \ref{th:finite}, we have $\max\tilde{\mathscr L}_{\cX} \psi(\vct x) = -\|\mH (\vct u - \vct u^*)\|^2\leq 0$. It was shown in \cite{tron2016bearing}[Proposition 3] that $\mH (\vct u - \vct u^*)$ equals zero if and only if $\vct u_{ij}=\vct u_{ij}^*$ for every $(i,j)\in\cE$. 
\end{proof}
As a result of this, a formation $\cF=(\cG,\vct x)$ with initial position $\vct x_0$ will converge to a formation $\vct x^\star$ which is similar to $\vct x^*$. If the formation is bearing rigid, $\vct x^\star$ is also similar to $\vct x^*$. Furthermore, following the same argument from Lemma \ref{lemma-centroid}, it can be shown that the centroid of Filippov solutions of \eqref{eq:dynamicsformationagg} is invariant.
%\TODO{Used Lyapunov function does not have bounded level sets and therefore agents may go to infinity, but because scale of the formation is decreasing this does not happen}

\subsection{Directed graphs}

In this section we consider the controller \eqref{eq:dynamicsformationagg} for directed sensing graphs, given by:

\begin{equation}
\dot{\vct x}_i = \sum_{j\in\mathcal{N}_i^+} \vct{u}_{ij} - \vct{u}_{ij}^*,
\label{eq:dynamicsformationdirected}
\end{equation}

which can be written in the aggregate form as:

\begin{equation}
\dot{\vct{x}}= \mH_+ (\vct u - \vct u^*).
\label{eq:formationdirected}
\end{equation}

We assume that each agent only acts based on the measurements directly obtained by itself. Similar to the directed consensus problem, each agent $i$ has its own private function $\psi_i(\vct x, \vct u^*)=\sum_{j\in\cN_i^+}\psi_{\{i,j\}}$ which tries to minimize thorough gradient descent. Evaluating the rate at which $\psi_i$ decreases is difficult since it is also dependent on the dynamics of neighbors of $i$. In the directed consensus problem, we were able to use the maximum distance between nodes as a global metric to measure how far the system is from equilibrium. For the problem at hand, finding a similar global metric seems unrealistic and the only option left is to investigate the evolution of private functions.

We begin by showing that if the sensing graph is a directed cycle, we can use $\psi(\vct x,\vct u^*)$ to prove stability of \eqref{eq:formationdirected}. Later we give intuition on the equilibria of $\psi_i$s and prove convergence of directed acyclic graphs.

\begin{proposition}
Controller \eqref{eq:formationdirected} is asymptotically stable for a directed cycle graph.
\label{prop:directedcycle}
\end{proposition}
\begin{proof}
In a directed cycle, we have $\dot{\vct x}_i=\vct u_{ij}-\vct u_{ij}^*$ where $j\in\cN_i^+$ is the only neighbor of $i$. Also, we have $\frac{\partial\psi}{\partial \vct x_i}=-(\vct u_{ij}-\vct u_{ij}^*)-(\vct u_{ik}-\vct u_{ik}^*)$ where $i\in\cN_k^+$. Assuming collisions do not occur, we have:
\begin{equation*} 
\begin{aligned}
\dot{\psi}&=\sum_{i\in\cV} -[\vct u_{ij}-\vct u_{ij}^*+\vct u_{ik}-\vct u_{ik}^*]^\trp(\vct u_{ij}-\vct u_{ij}^*)\\
&=\sum_{i\in\cV}( -\|\dot{\vct x}_i\|^2+\dot{\vct x}_i^\trp\dot{\vct x}_k)
\end{aligned}
\end{equation*}
which is due to $\dot{\vct x}_k=\vct u_{ki}-\vct u_{ki}^*$. Since there are as many edges as nodes, we can rewrite $\dot\psi$ over edges as:
\begin{equation*} 
\begin{aligned}
\dot\psi &= \sum_{(k,i)\in\cE} -\frac{1}{2}\|\dot{\vct x}_i\|^2+\dot{\vct x}_i^\trp\dot{\vct x}_k -\frac{1}{2}\|\dot{\vct x}_k\|^2 \\
&=\sum_{(k,i)\in\cE} -\frac{1}{2} \|\dot{\vct x}_i-\dot{\vct x}_k\|^2\leq 0
\end{aligned}
\end{equation*}
Hence $\dot\psi$ is always negative unless all nodes have the same velocity $\dot{\vct x}_i=\dot{\vct x}_k$. Suppose all $\dot{\vct x}_i = \vct w$, then we have $\vct u_{ij} -\vct w= \vct u_{ij}^*$. Taking the norm of both sides, we get $\vct w^\trp\vct u_{ij}=\frac{1}{2}\|\vct w\|^2$. Furthermore, we have $\sum_{i\in\cV}d_{ij}\vct u_{ij}=\vct 0$, hence taking a dot product with $\vct w$ we get $\sum_{i\in\cV}d_{ij}\vct w^\trp\vct u_{ij}=\sum_{i\in\cV}\frac{d_{ij}}{2}\|\vct w\|^2=\vct 0$ which means $\vct w=\vct 0$.
%$\mH_+=\mI_{nd}$ and $\mathfrak D \psi(\vct x) = -\mH (\vct u - \vct u^*) \bm\oplus \cI
%$ and $\cK[\cX](\vct x) = \mH_+ (\vct u - \vct u^*) \bm\oplus \cI^\prime = (\vct u - \vct u^*) \bm\oplus \cI^\prime$. Do the dot product and you will see that we get ... 
\end{proof}

When the out-degree of a node $i$ is one, as in a directed cycle graph, the equilibrium points of its objective function $\psi_i$ is a half-line that starts at the position of its neighbor and extends to infinity in the direction of $-\vct u_{ij}^*$. If the out-degree is more than one, the equilibrium point(s) of $\psi_i$ are such that $\sum_{j\in\cN_i^+}\vct u_{ij}=\sum_{j\in\cN_i^+}\vct u_{ij}^*$. This, however, does not necessarily mean that the bearing measurement of each neighbor $\vct u_{ij}$ is equal to the desired bearing $\vct u_{ij}^*$ assuming the equilibrium point(s) exists. Before we discuss the existence of equilibrium points, we present the the following definition which is motivated by this problem.

\begin{definition}[Bearing Persistence]
A directed graph $\cG$ is bearing persistent such that for any $\vct x$ and $\vct x^*\in\real{dn}$ and all $i\in\cV$, $\sum_{j\in\cN_i^+} \vct u_{ij}-\vct u_{ij}^*=\vct 0$ if and only if $\vct x$ and $\vct x^*$ are equivalent.
\end{definition}

\begin{remark}
A bearing persistent framework may not be baring rigid. The opposite direction is also true (see Fig. \ref{fig:persistence}). Also, It can be immediately deduced that undirected graphs and directed graphs with out-degree one are bearing persistent.  
\end{remark}

%Reverting or removing edge $(2,4)$ yields a bearing persistent framework.
%\TODO{argue that if maximum out-degree is two, then bearing persistence is achieved.}

\begin{figure}
\centering
\subfloat[]{\label{fig:persistence:1}\begin{tikzpicture}
\coordinate(n1) at (0,2);
\coordinate(n2) at (2,2);
\coordinate(n3) at (0,0);
\coordinate(n4) at (2,0);

\coordinate(ni) at (0,0);
\coordinate(nj) at (2,2);

\tikzstyle{arman}=[circle,draw,thick, inner sep=1pt, minimum size=3pt]

\node[arman](N1) at (n1) {1};
\node[arman](N2) at (n2) {2};
\node[arman](N3) at (n3) {3};
\node[arman](N4) at (n4) {4};

%\node[circle,draw,line width=0.8pt](NI) at (ni) {$i$};
%\node[circle,draw,line width=0.8pt](NJ) at (nj) {$j$};

%\draw[very thick] (NI) -- (NJ);

\begin{scope}[thick,decoration={
    markings,
    mark=at position 0.55 with {\arrow{latex}}}
    ] 
    \draw[postaction={decorate}] (N1)--(N2);
    \draw[postaction={decorate}] (N1)--(N3);
    \draw[postaction={decorate}] (N3)--(N4);
    \draw[postaction={decorate}] (N2)--(N4);
    \draw[postaction={decorate}] (N1)--(N4);
\end{scope}
\end{tikzpicture}} \quad
\subfloat[]{\label{fig:persistence:2}\begin{tikzpicture}
%\coordinate(n1) at (0,2);
\coordinate(n1) at (0.1632,2.25);
\coordinate(n2) at (3,2);
\coordinate(n3) at (0,0);
\coordinate(n4) at (3,0);

\coordinate(ni) at (0,0);
\coordinate(nj) at (2,2);

\tikzstyle{arman}=[circle,draw,thick, inner sep=1pt, minimum size=3pt]

\node[arman](N1) at (n1) {1};
\node[arman](N2) at (n2) {2};
\node[arman](N3) at (n3) {3};
\node[arman](N4) at (n4) {4};

%\node[circle,draw,line width=0.8pt](NI) at (ni) {$i$};
%\node[circle,draw,line width=0.8pt](NJ) at (nj) {$j$};

%\draw[very thick] (NI) -- (NJ);

\begin{scope}[thick,decoration={
    markings,
    mark=at position 0.55 with {\arrow{latex}}}
    ] 
    \draw[postaction={decorate}] (N1)--(N2);
    \draw[postaction={decorate}] (N1)--(N3);
    \draw[postaction={decorate}] (N3)--(N4);
    \draw[postaction={decorate}] (N2)--(N4);
    \draw[postaction={decorate}] (N1)--(N4);
\end{scope}
\end{tikzpicture}}\\
\subfloat[]{\label{fig:rigidity:1}\begin{tikzpicture}
\coordinate(n1) at (0,2);
\coordinate(n2) at (2,2);
\coordinate(n3) at (0,0);
\coordinate(n4) at (2,0);

\coordinate(ni) at (0,0);
\coordinate(nj) at (2,2);

\tikzstyle{arman}=[circle,draw,thick, inner sep=1pt, minimum size=3pt]

\node[arman](N1) at (n1) {1};
\node[arman](N2) at (n2) {2};
\node[arman](N3) at (n3) {3};
\node[arman](N4) at (n4) {4};

%\node[circle,draw,line width=0.8pt](NI) at (ni) {$i$};
%\node[circle,draw,line width=0.8pt](NJ) at (nj) {$j$};

%\draw[very thick] (NI) -- (NJ);

\begin{scope}[thick,decoration={
    markings,
    mark=at position 0.55 with {\arrow{latex}}}
    ] 
    \draw[postaction={decorate}] (N1)--(N2);
    \draw[postaction={decorate}] (N2)--(N4);
    \draw[postaction={decorate}] (N4)--(N3);
    \draw[postaction={decorate}] (N3)--(N1);
    \draw[postaction={decorate}] (N1)--(N4);
\end{scope}
\end{tikzpicture}}\quad
\subfloat[]{\label{fig:rigidity:2}\begin{tikzpicture}
\coordinate(n1) at (0,2);
\coordinate(n2) at (3,2);
\coordinate(n3) at (0,0);
\coordinate(n4) at (3,0);

\coordinate(ni) at (0,0);
\coordinate(nj) at (2,2);

\tikzstyle{arman}=[circle,draw,thick, inner sep=1pt, minimum size=3pt]

\node[arman](N1) at (n1) {1};
\node[arman](N2) at (n2) {2};
\node[arman](N3) at (n3) {3};
\node[arman](N4) at (n4) {4};

%\node[circle,draw,line width=0.8pt](NI) at (ni) {$i$};
%\node[circle,draw,line width=0.8pt](NJ) at (nj) {$j$};

%\draw[very thick] (NI) -- (NJ);

\begin{scope}[thick,decoration={
    markings,
    mark=at position 0.55 with {\arrow{latex}}}
    ] 
    \draw[postaction={decorate}] (N1)--(N2);
    \draw[postaction={decorate}] (N2)--(N4);
    \draw[postaction={decorate}] (N4)--(N3);
    \draw[postaction={decorate}] (N3)--(N1);
\end{scope}
\end{tikzpicture}}
\caption{In \protect\subref{fig:persistence:1} and \protect\subref{fig:persistence:2} sum of the bearing measurements of nodes with the same index is equal but the formations are not equivalent, which means the underlying graph is not bearing persistent. \protect\subref{fig:rigidity:1} and  \protect\subref{fig:rigidity:2} are bearing persistent graphs. Graph in \protect\subref{fig:rigidity:1} is bearing rigid as well while \protect\subref{fig:rigidity:2} is not.}
\label{fig:persistence}
\end{figure}
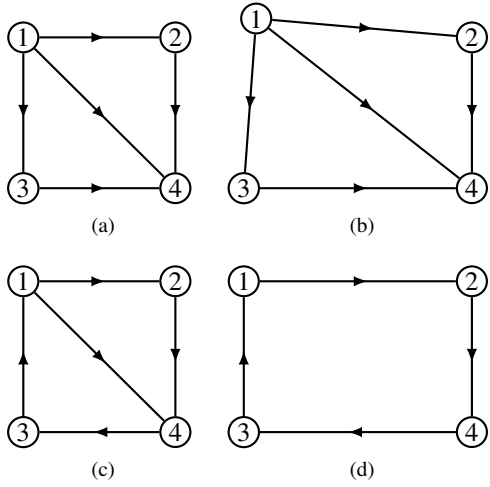

Even if the sensing graph is not bearing persistent, it is not trivial to study the equilibria of \eqref{eq:formationdirected}. In some applications, achieving the exact bearings between the agents might not be important, but rather the overall placement of an agent with respect to those it observes is. Here we present a short and informal proof on uniqueness of equilibrium of \eqref{eq:dynamicsformationdirected}.

The equilibrium point of \eqref{eq:dynamicsformationdirected} for agent $i$ with $|\cN_i^+|>1$ is a point such that $\sum_{j\in\cN_i^+} \vct u_{ij}=\sum_{j\in\cN_i^+}\vct u_{ij}^*=\vct v^*$. If $\|\vct v^*\|=|\cN_i^+|$, then $\vct x_i\rightarrow\infty$ if neighbors of $i$ are not all coincident. Hence we assume that always $\|\vct v^*\|<|\cN_i^+|$, or the given desired bearings for an agent are not collinear. Controller \eqref{eq:dynamicsformationdirected} steers $i$ to a point where the sum of its bearing measurements equals $\vct v^*$. The following definition is motivated by this behavior.

\begin{figure}
\centering
\input{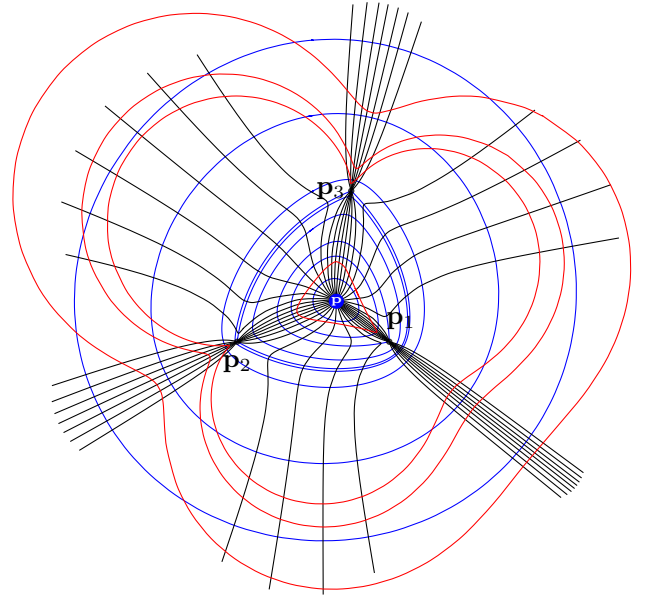}
\caption{Multiple concentric 3-ellipses (blue curves) with foci $\{\vct p_i\}_{i=1}^3$. Point $\vct p$ is the geometric median of focal points. Direction of gradient of $\vartheta(.)$ does not change along each black curve starting from $\vct p$, and its magnitude does not change along red curve.}
\label{fig:ellipse}
\end{figure}

\begin{definition}
a \emph{$k$-ellipsoid} is the set of points such that sum of their euclidean distances from $k$ fixed points $\{\vct p_i\in\real{d}\}_{i=1}^k$ called foci is constant. Let $\vartheta(\vct y) \triangleq \sum_{i=1}^k \|\vct y-\vct p_i\|$ be the sum of distances to foci from point $\vct y$. A $k$-ellipsoid denoted as $\Upsilon(\rho)$ is the boundary of the set-valued map $\Theta(\rho)=\{\vct y\in\real{d} \; | \; \vartheta(\vct y) \leq\rho\}$ for a given $\rho\geq\rho^\star$ where $\rho^\star=\min_{\vct y}\vartheta(\vct y)$.
\end{definition}

$\Theta(\rho)$ is a sublevel set of of a convex function and is therefore a bounded convex set. $\Upsilon$ is a closed convex surface and is smooth if it does not contain any of the focal points \cite{nie2008semidefinite}.

Let $\vct v_i\triangleq \frac{\vct p_i-\vct y}{\|\vct p_i-\vct y\|}$ be the unit vector pointing towards $\vct p_i$ from $\vct y$. Gradient of the function $\vartheta(\vct y)$ at a point $\vct y\neq\vct p_i$ is given by $\frac{\partial\vartheta}{\partial\vct y}=\sum_{i=i}^k-\vct v_i$ and its Hessian is given by $\frac{\partial^2\vartheta}{\partial\vct y^2}=\sum_{i=i}^k\frac{1}{\|\vct p_i-\vct y\|}\mP(\vct v_i)$. Hessian of $\vartheta(\vct y)$ is positive definite unless foci are collinear. Even if that is the case, it can easily be shown that $\vartheta(\vct y)$ is strictly convex along any line except the line that contains the foci.

Using this fact, it can be argued that $\Upsilon(\rho)$ for $\rho>\rho^\star$ does not contain a line segment and the direction of gradient of $\vartheta(\vct y)$ or $\sum_{i=1}^k-\vct v_i$ which is parallel to the tangent hyperplane of $\Upsilon(\rho)$ is unique on $\Upsilon(\rho)$. Furthermore, at the geometric median (or line segment) $\|\sum_{i=1}^k\vct v_i\|$ is zero but as $\|\vct y\|\rightarrow\infty$ we have $\|\sum_{i=1}^k\vct v_i\|\rightarrow |\cN_i^+|$. Due to convexity of $\vartheta(.)$, $\mathfrak D\vartheta(\vct y)$ must attain any direction and any length between zero and $|\cN_i^+|$ due to being a monotone function \cite{kachurovskii1960monotone} (see Fig. \ref{fig:ellipse}).

Having established uniqueness of the equilibrium point, it is straightforward to prove stability of \eqref{eq:formationdirected} for directed acyclic graphs. Leaves of a directed acyclic graph does not have any neighbors and are stationary. We define the \emph{degree of cascade} of a node to be the length of the longest path from that node to a leaf of the graph and is unique due to absence of cycles. Starting from degree one to higher degrees, nodes reach their equilibrium. 

\begin{figure*}[t]
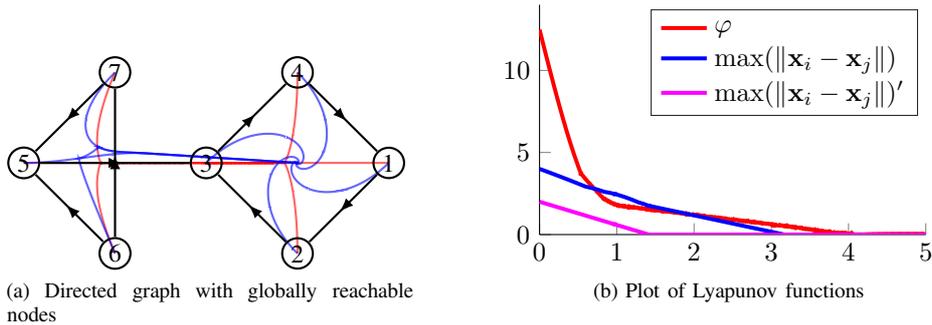

\centering
\subfloat[Directed graph with globally reachable nodes]{\label{fig:cons:1}\input{fig_cons1.tex}}\quad\quad\quad
\subfloat[Plot of Lyapunov functions]{\label{fig:cons:2}\input{fig_cons1_plot.tex}}
\caption{The red points correspond to the undirected version of the  graph in \protect\subref{fig:cons:1} and the blue points correspond to the directed graph. The magenta plot in \protect\subref{fig:cons:2} corresponds to the strongly connected component 1-2-3-4. The proposed upper-bound on convergence time $\frac{l}{4}\sec^2{\frac{\pi}{4}}$ is $\sqrt{2}$ which is exact in this case for the strongly connected component 1-2-3-4.}
\label{fig:conssimul}
\end{figure*}

\begin{figure*}[t]
\centering
\subfloat[Desired formation]{\label{fig:forms:1}\begin{tikzpicture}[thick,scale=1.2, every node/.style={transform shape}]
\tikzstyle{arman}=[circle,draw,thick, inner sep=1pt, minimum size=1pt]
\coordinate(n1) at (0,1.543);
\coordinate(n2) at (-1.2344,0.61721);
\coordinate(n3) at (1.2344,0.61721);
\coordinate(n4) at (0,-0.92582);
\node[arman](N1) at (n1){\footnotesize 1};
\node[arman](N2) at (n2){\footnotesize 2};
\node[arman](N3) at (n3){\footnotesize 3};
\node[arman](N4) at (n4){\footnotesize 4};
\begin{scope}[thick,decoration={markings,mark=at position 0.55 with {\arrow{latex}}}]
\draw[postaction={decorate}] (N1)--(N2);
\draw[postaction={decorate}] (N2)--(N4);
\draw[postaction={decorate}] (N4)--(N3);
\draw[postaction={decorate}] (N3)--(N1);
\draw[postaction={decorate}] (N1)--(N4);
\end{scope}
\end{tikzpicture}} \quad
\subfloat[Initial formation and trajectories]{\label{fig:forms:2}\input{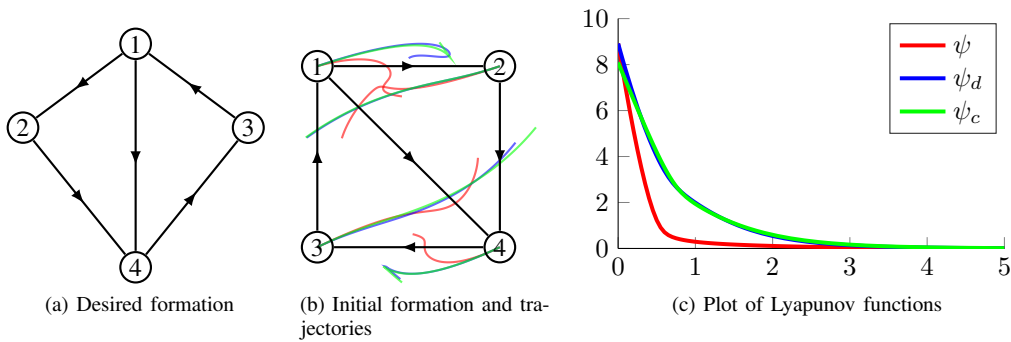}}\quad
\subfloat[Plot of Lyapunov functions]{\label{fig:forms:3}% This file was created by matlab2tikz.
\begin{tikzpicture}

\begin{axis}[%
width=2in,
height=1.2in,
at={(0in,0.1in)},
scale only axis,
xmin=0,
xmax=5,
ymin=0,
ymax=10,
axis background/.style={fill=white},
axis x line*=bottom,
axis y line*=left,
legend style={legend cell align=left, align=left, draw=white!15!black}
]
\addplot [color=red, line width=1.5pt]
  table[row sep=crcr]{%
0	8.91607960096898\\
0.05	7.90383318744505\\
0.1	6.93138254591837\\
0.15	6.00185005901714\\
0.2	5.11924990659049\\
0.25	4.28880536283653\\
0.3	3.51731395487234\\
0.35	2.81349261325567\\
0.4	2.18816626851594\\
0.45	1.65410369419366\\
0.5	1.22499328402654\\
0.55	0.910945671742685\\
0.6	0.705940600546193\\
0.65	0.580293652842925\\
0.7	0.499789797368673\\
0.75	0.443126294819634\\
0.8	0.400004081644163\\
0.85	0.365448080906375\\
0.9	0.336796711292426\\
0.95	0.312464897749409\\
1	0.291427975997776\\
1.05	0.272982995417578\\
1.1	0.256625550693032\\
1.15	0.241980599015837\\
1.2	0.228761178150248\\
1.25	0.216742617622251\\
1.3	0.205745787579953\\
1.35	0.195625818526856\\
1.4	0.186264240594933\\
1.45	0.177563327078946\\
1.5	0.169441903685867\\
1.55	0.16183216345726\\
1.6	0.154677193385607\\
1.65	0.147929019712731\\
1.7	0.141547041531545\\
1.75	0.135496762010768\\
1.8	0.129748752319662\\
1.85	0.124277800470402\\
1.9	0.11906220901266\\
1.95	0.114083213749396\\
2	0.109324501591398\\
2.05	0.104771810079511\\
2.1	0.100412594454756\\
2.15	0.0962357507566415\\
2.2	0.0922313854838603\\
2.25	0.0883906239979041\\
2.3	0.0847054511853642\\
2.35	0.0811685789874322\\
2.4	0.0777733363055989\\
2.45	0.0745135775383645\\
2.5	0.0713836066237536\\
2.55	0.0683781139791151\\
2.6	0.0654921241610266\\
2.65	0.0627209524285788\\
2.7	0.060060168694711\\
2.75	0.0575055676023274\\
2.8	0.055053143672702\\
2.85	0.0526990706498954\\
2.9	0.0504396843121622\\
2.95	0.0482714681443229\\
3	0.046191041367749\\
3.05	0.0441951489102638\\
3.1	0.0422806529696797\\
3.15	0.040444525884197\\
3.2	0.0386838440724365\\
3.25	0.0369957828471106\\
3.3	0.0353776119406276\\
3.35	0.0338266916094216\\
3.4	0.0323404692074676\\
3.45	0.0309164761390897\\
3.5	0.0295523251174623\\
3.55	0.0282457076687111\\
3.6	0.026994391832701\\
3.65	0.0257962200208539\\
3.69999999999999	0.0246491069989844\\
3.74999999999999	0.0235510379694515\\
3.79999999999999	0.0225000667321285\\
3.84999999999999	0.0214943139079717\\
3.89999999999999	0.0205319652124804\\
3.94999999999999	0.0196112697692191\\
3.99999999999999	0.0187305384559183\\
4.04999999999999	0.0178881422775803\\
4.09999999999999	0.0170825107625583\\
4.14999999999999	0.0163121303788212\\
4.19999999999999	0.0155755429686055\\
4.24999999999999	0.0148713442004388\\
4.29999999999999	0.0141981820381345\\
4.34999999999999	0.0135547552268228\\
4.39999999999999	0.0129398117964406\\
4.44999999999999	0.0123521475833582\\
4.49999999999999	0.0117906047710024\\
4.54999999999999	0.0112540704504486\\
4.59999999999999	0.0107414752020217\\
4.64999999999999	0.0102517916989657\\
4.69999999999999	0.00978403333423378\\
4.74999999999999	0.00933725287141364\\
4.79999999999999	0.00891054112074825\\
4.84999999999999	0.0085030256411402\\
4.89999999999999	0.00811386946894803\\
4.94999999999999	0.00774226987429324\\
5	0.0073874571455036\\
};
\addlegendentry{$\psi$}

\addplot [color=blue, line width=1.5pt]
  table[row sep=crcr]{%
0	8.91607960096898\\
0.0284420231278838	8.57079251135209\\
0.0568840462557676	8.23294389729056\\
0.0853260693836513	7.90266931271523\\
0.170057393252485	6.96542384007019\\
0.254788717121319	6.09935775637897\\
0.339520040990153	5.30662511888036\\
0.407561369361813	4.72476959274945\\
0.452438883911975	4.36901423813812\\
0.488839376531556	4.09786449938059\\
0.525239869151138	3.84322626848403\\
0.561640361770719	3.60602362326552\\
0.598040854390301	3.38719963328228\\
0.652428272721119	3.09606262112308\\
0.706815691051937	2.84808669752553\\
0.761203109382754	2.63960793606046\\
0.815590527713572	2.462735447913\\
0.86997794604439	2.30817413689529\\
0.930032128036923	2.15436961322514\\
0.990086310029456	2.01228285898408\\
1.05014049202199	1.87860144838444\\
1.11019467401452	1.75163417636866\\
1.17024885600705	1.63065225396395\\
1.23030303799959	1.51552372311448\\
1.31103948961137	1.37016230345981\\
1.39177594122314	1.23574595709186\\
1.47251239283492	1.11220633491038\\
1.5532488444467	0.999282892158877\\
1.63398529605848	0.896572943890638\\
1.71472174767026	0.803551232625015\\
1.86150235024913	0.657226134033224\\
2.008282952828	0.536806854271108\\
2.15506355540687	0.438242170964007\\
2.30184415798574	0.357848771373553\\
2.44862476056461	0.292419867283114\\
2.59540536314348	0.239216387311805\\
2.81883580663849	0.17667563242616\\
3.04226625013351	0.130945716213665\\
3.26569669362852	0.0974179985954235\\
3.48912713712353	0.0727652214524889\\
3.71255758061854	0.0545828705961389\\
3.93598802411356	0.0411262831827359\\
4.34465404710142	0.0248031035213491\\
4.75332007008929	0.0152123525530437\\
5	0.0114291029843747\\
};
\addlegendentry{$\psi_d$}

\addplot [color=green, line width=1.5pt]
  table[row sep=crcr]{%
0	8.08765247622279\\
0.0465606631094765	7.70968224926614\\
0.093121326218953	7.33241951539807\\
0.13968198932843	6.95628315504631\\
0.252329460413069	6.0567976538646\\
0.364976931497708	5.18013949515016\\
0.447939256356094	4.55891825128925\\
0.508656075353339	4.12438050484453\\
0.549060584598957	3.84784527343093\\
0.562271600435405	3.76002002046277\\
0.575482616271853	3.67363130894141\\
0.588693632108302	3.58875892199571\\
0.60190464794475	3.5054868652044\\
0.636639417210236	3.29483193719769\\
0.671374186475723	3.09745230347895\\
0.706108955741209	2.91485752623106\\
0.740843725006695	2.7482422745432\\
0.775578494272182	2.59818447543117\\
0.828833326051885	2.39935814502548\\
0.882088157831589	2.2329327171779\\
0.935342989611293	2.09044658287886\\
0.988597821390996	1.96436756430557\\
1.0418526531707	1.84959378803983\\
1.10909333343818	1.71596433912055\\
1.17633401370565	1.59203108282812\\
1.24357469397313	1.47608302126028\\
1.3108153742406	1.36722428804954\\
1.37805605450808	1.26498888977779\\
1.48034105032024	1.12151842493551\\
1.58262604613239	0.992044325865349\\
1.68491104194455	0.875800731495779\\
1.7871960377567	0.771927537933034\\
1.88948103356886	0.679495870931586\\
2.06291376142742	0.546149204393804\\
2.23634648928599	0.438270231376454\\
2.40977921714456	0.351493206044586\\
2.58321194500312	0.281927268127379\\
2.75664467286169	0.226260864533954\\
2.98215538946814	0.170109662056497\\
3.2076661060746	0.128126395895114\\
3.43317682268105	0.0967183402523983\\
3.65868753928751	0.0731798470940831\\
4.05134876201956	0.0453307597260164\\
4.44400998475162	0.0282883640770186\\
4.83667120748368	0.0177572925249714\\
5	0.0146429434987629\\
};
\addlegendentry{$\psi_c$}

\end{axis}

\begin{axis}[%
width=2.2in,
height=1.4in,
at={(0in,0in)},
scale only axis,
xmin=0,
xmax=1,
ymin=0,
ymax=1,
axis line style={draw=none},
ticks=none,
axis x line*=bottom,
axis y line*=left,
legend style={legend cell align=left, align=left, draw=white!15!black}
]
\end{axis}
\end{tikzpicture}%
}
\caption{ Trajectories of an undirected graph (red), a directed graph (blue), and the cycle graph 1-2-4-3 (green).}
\label{fig:formssimul}
\end{figure*}

For the case of directed graphs with cycles, proving stability still remains a challenge. One natural first step could be to see if equilibrium points of \eqref{eq:formationdirected} are Hurwitz-stable with respect to perturbations. Jacobian matrix of \eqref{eq:formationdirected} is given by $\mH_+\mR_B|_{\vct x^*}$ or -$\mH_+\diag(\frac{1}{d_{ij}^*}\mP(\vct u_{ij}^*))\mH^\trp$, where $\mR_B$ is called the \emph{bearing rigidity matrix}. This matrix is very similar to the \emph{directed bearing Laplacian matrix} $\mL_B=\mH_+\diag(\mP(\vct u_{ij}^*))\mH^\trp$ defined in \cite{zhao2015bearing}. For the graph given in Fig. \ref{fig:rigidity:1} with positions $\vct x_1=[0,0]^\trp$, $\vct x_2=[2,0]^\trp$, $\vct x_3=[3,-4]^\trp$, and $\vct x_4=[2,-2]^\trp$, Jacobian matrix of \eqref{eq:formationdirected} and $-\mL_B$ both have an eigenvalue with a positive real part, which rejects the conjecture made in \cite{zhao2015bearing} on bearing Laplacian matrix having eigenvalues with nonnegative real parts.

\section{SIMULATION RESULTS}
In this section, we present simulation results for the both bearing-only consensus and formation control problems. In Fig. \ref{fig:conssimul}, the trajectory of an undirected and directed graph with the same vertices is given for the consensus problem. In Fig. \ref{fig:formssimul}, trajectories of an undirected graph, a strongly connected graph and a directed cycle graph is presented for the formation control problem.

\section{CONCLUSIONS}
%\TODO{Scale the controller output of each edge by a constant.}
We presented stability results for the bearing-only consensus and formation control problems. There are remaining problems which need further attention. In the consensus problem of strongly connected directed graphs, finite-time convergence remains unsolved. Also, bearing-only formation control in cyclic directed graphs is not addressed yet and the notion of bearing persistence needs more study in the future.
%In this paper we provided new theoretical insights on partitions of rigid components, and presented a greedy algorithm for rigidity recovery which adds the minimum number of required edges. There is a stark contrast between solutions for localization and formation control which are relatively well studied and distributed, and algorithms for rigid component identification and rigidity recovery which are relatively new and centralized. In our future work, we will study distributed variants of the latter. Moreover, we will also work on extending our theoretical conditions on the minimum number of edges necessary for rigidity recovery to the 3-D case. Other interesting venues for future work include the study of optimality criteria for choosing edges, such as algebraic connectivity, or quantities related to the robustness to noise or the integration with distributed coverage problem.


%%%%%%%%%%%%%%%%%%%%%%%%%%%%%%%%%%%%%%%%%%%%%%%%%%%%%%%%%%%%%%%%%%%%%%%%%%%%%%%%

%%%%%%%%%%%%%%%%%%%%%%%%%%%%%%%%%%%%%%%%%%%%%%%%%%%%%%%%%%%%%%%%%%%%%%%%%%%%%%%%



% \section*{ACKNOWLEDGMENT}

\bibliographystyle{ieee}
\bibliography{bibroot}
%%%%%%%%%%%%%%%%%%%%%%%%%%%%%%%%%%%%%%%%%%%%%%%%%%%%%%%%%%%%%%%%%%%%%%%%%%%%%%%%

\end{document}